\documentclass[reqno]{amsart}     
\usepackage{pstricks,pst-node,pst-coil,pst-plot}

\theoremstyle{plain}
\newtheorem{step}{Step}
\newtheorem{case}{Case}
\newtheorem{thm}{Theorem}[section]
\newtheorem{theorem}[thm]{Theorem}
\newtheorem{cor}[thm]{Corollary}

\newtheorem{lem}[thm]{Lemma}
\newtheorem{lemma}[thm]{Lemma}
\newtheorem{prop}[thm]{Proposition}

\theoremstyle{remark}

\newtheorem{remark}[thm]{Remark}

\theoremstyle{definition}

\def\al{{\alpha}}

\def\si{{\sigma}}

\def\ep{{\epsilon}}

\def\th{{\theta}}

\def\ph{{\varphi}}
\def\phi{{\varphi}}

\DeclareMathAlphabet{\doba}{U}{msb}{m}{n}

\gdef\mR{\doba{R}}

\def\Vol{{\mathop{\rm vol}}}     
\let\vol\Vol
\def\Ric{{\mathop{\rm Ric}}}

\let\tr\trace

\def\divergence{{\mathop{\rm div}}}

\def\gammatil{\widetilde{\gamma}}

\let\<\langle 
\let\>\rangle

\newcommand{\definedas}{\mathrel{\raise.095ex\hbox{\rm :}\mkern-5.2mu=}}

\begin{document}


\title [Solvability of the vacuum Einstein constraint equations] {A 
limit equation associated to the solvability of the vacuum Einstein
  constraint equations using the conformal method}

\author{Mattias Dahl}
\address{Institutionen f\"or Matematik \\
  Kungliga Tekniska H\"ogskolan \\
  100 44 Stockholm \\
  Sweden} \email{dahl@math.kth.se}

\author{Romain Gicquaud}
\address{Laboratoire de Math\'ematiques et de Physique Th\'eorique \\
  UFR Sciences et Technologie \\
  Facult\'e Fran\c cois Rabelais \\
  Parc de Grandmont \\
  37200 Tours \\
  France} \email{romain.gicquaud@lmpt.univ-tours.fr}

\author{Emmanuel Humbert}
\address{Institut \'Elie Cartan, BP 239 \\
  Universit\'e de Nancy 1 \\
  54506 Vandoeuvre-l\`es-Nancy Cedex \\
  France} \email{humbert@iecn.u-nancy.fr}

\begin{abstract}
  Let $(M,g)$ be a compact Riemannian manifold on which a trace-free
  and divergence-free $\sigma \in W^{1,p}$ and a positive function
  $\tau \in W^{1,p}$, $p > n$, are fixed. In this paper, we study the
  vacuum Einstein constraint equations using the well known conformal
  method with data $\sigma$ and $\tau$. We show that if no
  solution exists then there is a non-trivial solution of another
  non-linear limit equation on $1$-forms. This last equation can be
  shown to be without solutions no solution in many situations. As a
  corollary, we get existence of solutions of the vacuum Einstein
  constraint equation under explicit assumptions which in particular
  hold on a dense set of metrics $g$ for the $C^0$-topology.
\end{abstract}

%
%
%


\date{January 28, 2011}


\maketitle

\tableofcontents

\section{Introduction}

\subsection{Background}

The Cauchy problem in general relativity asks for a space-time
development of a given initial data set $(M,g,K)$ consisting of an
$n$-dimensional manifold $M$ equipped with a Riemannian metric $g$ and
a symmetric $(0,2)$-tensor $K$. Such a development is a manifold
$\mathcal{M} \definedas M \times \mR$ equipped with a globally
hyperbolic Lorentzian metric $G$ satisfying Einstein's equation
\begin{equation*}
  \Ric^G - \frac{1}{2} R^G G = 8 \pi T.
\end{equation*}
Here $\Ric^G$ and $R^G$ are the Ricci and scalar curvatures of the
metric $G$ and $T$ is the energy-momentum tensor modeling the matter.
For $(\mathcal{M}, G)$ to be a development it is required that $(M,
g)$ embeds isometrically on the slice $M \times \{0 \}$ with second
fundamental form $K$.

The initial data $(M,g,K)$ cannot be freely specified but must itself
satisfy the so-called the {\em Einstein constraint equations}.  In
this paper, we are interested in the {\em Vacuum Einstein constraint
  equations}
\begin{subequations}
  \begin{align}
    R - |K|^2 + (\tr K)^2 &= 0,
    \label{hamiltonian0} \\
    \divergence K - d \tr K &= 0.
    \label{momentum0}
  \end{align}
\end{subequations}
corresponding to the case where the energy-momentum tensor $T$
vanishes everywhere. Y. Choquet-Bruhat and R. Geroch \cite{CBG69}
proved that such initial data give rise to a unique development
$(\mathcal{M},G)$ as described above.

The most efficient method to find initial data satisfying the vacuum
constraint equations is the conformal method developed by A.
Lichnerowicz \cite{Li44} and Y. Choquet-Bruhat and J. W. York
\cite{CBY80}. For this approach we let $(M,g)$ be a Riemannian
manifold, $\tau$ a smooth function and $\sigma$ be a trace-free and
divergence-free symmetric $(0,2)$-tensor on $M$. We consider the
following system of equations for $\phi$ and $W$,
\begin{subequations}
  \begin{align}
\frac{4(n-1)}{n-2} \Delta \phi + R \phi 
&= 
-\frac{n-1}{n} \tau^2 \phi^{N-1} + |\sigma + LW|^2 \phi^{-N-1},
    \label{hamiltonian} \\
-\frac{1}{2} L^* L W 
&= 
\frac{n-1}{n} \phi^N d\tau,
    \label{momentum}
  \end{align}
\end{subequations}
where $\phi$ is a positive smooth and where $W$ is a smooth one-form.
In this system the first equation is called the {\em Lichnerowicz
  equation}, and the second equation is called the {\em vector
  equation}. These coupled equations are together called the {\em
  conformal constraint equations}. We have set $N \definedas
\frac{2n}{n-2}$. Further, $\Delta$ is the non-negative Laplacian
acting on functions, $R$ the scalar curvature, and $L$ the conformal
Killing operator, all defined using the metric $g$. Except when the
notation indicates otherwise all metric-dependent objects on $M$ are
assumed to be defined using the metric $g$. The conformal Killing
operator is defined on a $1$-form $V$ by taking the symmetric,
trace-free part of $\nabla V$. The vector fields corresponding to
$1$-forms in the kernel of $L$ are precisely the conformal Killing
vector fields, to assume that $(M,g)$ has no conformal Killing vector
fields is thus equivalent to assume that $\ker L = 0$.

If $(\phi,W)$ are solutions of the conformal constraint equations
\eqref{hamiltonian}-\eqref{momentum} then
\begin{equation*}
  \left( M, \phi^{N-2} g, 
    \frac{\tau}{n} \phi^{N-2} g + \phi^{-2} (\si + LW) \right)
\end{equation*}
is an initial data set satisfying the vacuum Einstein's constraint
equations \eqref{hamiltonian0}-\eqref{momentum0}, see for example
\cite{BartnikIsenberg} or \cite[Chapter VII]{ChoquetBruhat}.  Note
that the function $\tau$ is then the mean curvature of the embedding
of $(M,g)$ into $(\mathcal{M},G)$.

Let us now describe the results on existence of solutions to the
system \eqref{hamiltonian}-\eqref{momentum} which are known until now.
We restrict to the case of closed manifolds $M$ . The case of constant
functions $\tau$ is completely understood, see \cite{Isenberg}, while
the case of non-constant $\tau$ in its full generality is still open.
All the known existence results require that $d \tau /\tau$ is small
or that $\sigma$ is small. The reader can refer for instance to
\cite{IM92}, \cite{ACI08}, or \cite{MaxwellNonCMC}. In the presence of
matter, then again some smallness of the data is needed, see
\cite{HNT1}. Some results are known for initial data of lower
regularity, see \cite{HNT2}, \cite{Maxwellrough}.

The main result of this paper, Theorem \ref{main}, is that if the
conformal constraint equations \eqref{hamiltonian}-\eqref{momentum}
have no solution, then there exists a non-trivial $1$-form $W$ solving
the equation
\begin{equation} \label{eqLimit} -\frac{1}{2} L^*L W = \al_0
  \sqrt{\frac{n-1}{n}} |LW| \frac{d\tau}{\tau}
\end{equation}
for some $\al_0 \in (0,1]$.  This result has advantages compared to
previous results:
\begin{itemize}
\item In the case that Equation \eqref{eqLimit} has no solution for
  any $\al_0 \in (0,1]$, then we obtain that there exists a solution
  of the conformal constraint equations
  \eqref{hamiltonian}-\eqref{momentum}. We also obtain the compactness
  of the set of solutions.
\item In the case that Equation \eqref{eqLimit} has no solution,
  $g,\tau$ are in the interior of the set of data for which
  \eqref{hamiltonian}-\eqref{momentum} have solutions for any $\sigma$,
  see Proposition \ref{interior_data}.
\item If $d \tau /\tau$ is small enough, then Equation \eqref{eqLimit}
  has no solution. The point here is that the smallness condition is
  explicit, see Corollaries \ref{Cor1} and \ref{Cor2}.
\item We do not need to assume that the initial data $g,\tau,\sigma$
  are smooth.
\item This is maybe the main point: the assumptions required in
  Corollary \ref{Cor1} hold for a dense set of metrics with respect to
  the $C^0$-topology. In other words, let $\tau$ be fixed with
  constant sign. Then, there exists a dense set of metrics with
  respect to the $C^0$-topology such that for all $\sigma$, we get
  solutions of Equations \eqref{hamiltonian}-\eqref{momentum}, see
  Corollary \ref{Cor3}. To our knowledge, this is the first result
  showing that the set of solutions of the conformal constraint
  equations is large. Furthermore, it was not known that for any such
  $\tau$, we can find even one metric giving rise to solutions of the
  conformal constraint equations.
\item Theorem \ref{main} lets us expect some further developments by
  finding other situations where Equation \eqref{eqLimit} has no
  solution.
\end{itemize}

The method we use to prove our main result is to decrease the exponent
$N$ of $\phi$ in the right hand side of the vector equation
\eqref{momentum} to $N-\epsilon$, this idea was introduced to us by J.
Isenberg.  The new equations obtained in this way we call
``subcritical'' (even though the situation here is not related to
criticality of Sobolev embedding as for the Yamabe problem).  At a
certain point of the proof we can find a supersolution for the
subcritical equations regardless of the size of the coefficients in
the equations, this makes it possible to solve the subcritical
equations for all $\ep > 0$. In \eqref{defEnergy} below we define an
``energy'' for solutions to the subcritical equations. It turns out
that the behavior of solutions as $\ep \to 0$ is controlled by this
energy, if it is bounded we get convergence to a solution of the
conformal constraint equations, otherwise certain rescaled solutions
converge to a solution to the limit equation.

There are only few results concerning non-existence of solutions of
the conformal constraint equations. In \cite[Theorem
2]{IsenbergOMurchadha} it is shown that there are no solutions if
$\sigma \equiv 0$ and $d\tau /\tau$ is small enough. In Theorem
\ref{nonexist} we strengthen this result by giving a more explicit
bound on $d \tau / \tau$.

\subsection{Statement of results}

Let $M$ be a compact manifold of dimension $n$, our goal is to find
solutions to the vacuum Einstein constraint equations using the
conformal Method. The given data on $M$ consists of
\begin{equation} \label{givendata}
  \begin{minipage}[c]{11cm} 
    \begin{itemize}
    \item a Riemannian metric $g \in W^{2, p}$,
    \item a function $\tau \in W^{1, p}$,
    \item a symmetric, trace- and divergence-free $(0,2)$-tensor
      $\sigma \in W^{1, p}$,
    \end{itemize}
  \end{minipage}
\end{equation}
with $p > n$, and one is required to find
\begin{equation*}
  \begin{minipage}[c]{11cm} 
    \begin{itemize}
    \item a positive function $\phi \in W^{2, p}$,
    \item a 1-form $W \in W^{2, p}$,
    \end{itemize}
  \end{minipage}
  \tag*{\phantom{(4)}}
\end{equation*}
which satisfy the conformal constraint equations
\eqref{hamiltonian}-\eqref{momentum}.

We use standard notation for function spaces, such as $L^p$, $C^k$,
and Sobolev spaces $W^{k,p}$. It will be clear from the context if the
notation refers to a space of functions on $M$, or a space of sections
of some bundle over $M$. For spaces of functions the subscript $+$ is
used to indicate the subspace of positive functions.

In our existence result we will assume that
\begin{equation} \label{assumptions}
  \begin{minipage}[c]{11cm} 
    \begin{itemize}
    \item $\tau$ vanishes nowhere,
    \item $(M,g)$ has no conformal Killing vector fields,
    \item $\sigma \not\equiv 0$ if $Y(g) \geq 0$.
    \end{itemize}
  \end{minipage}
\end{equation}
Here $Y(g)$ is the Yamabe constant of the conformal class of $g$, that
is
\begin{equation*}
  Y(g) \definedas 
  \inf 
  \frac{\int_M R^{\widetilde{g}} \, dv^{\widetilde{g}}}
  {\Vol^{\widetilde{g}} (M)^{\frac{n-2}{n}}} ,
\end{equation*} 
where the infimum is taken over all $\widetilde{g} \in W^{2,p}$
conformal to $g$. Our main result is the following Theorem.

\begin{theorem} \label{main} Let data be given on $M$ as specified in
  \eqref{givendata} and assume that \eqref{assumptions} holds. Then,
  at least one of following assertions is true.
  \begin{itemize}
  \item The system of equations \eqref{hamiltonian}-\eqref{momentum}
    admits a solution $(\phi,W)$ with $\phi > 0$.  Furthermore, the
    set of solutions $(\phi,W) \in W^{2, p}_+ \times W^{2,p}$ is
    compact.
  \item There exists a non-trivial solution $W \in W^{2, p}$ of the
    equation
    \begin{equation} \label{eqLimit1} -\frac{1}{2} L^*L W = \al_0
      \sqrt{\frac{n-1}{n}} |LW| \frac{d\tau}{\tau}
    \end{equation}
    for some $\al_0 \in (0,1]$.
  \end{itemize}
\end{theorem}

We will call Equation \eqref{eqLimit1} the {\em limit equation} since
it appears as the equation satisfied by a limit of rescaled solutions.
As an immediate consequence we conclude that if Equation
\eqref{eqLimit1} has no solution for any $\al_0 \in (0,1]$, then the
conformal constraint equations \eqref{hamiltonian}-\eqref{momentum}
admit a solution $(\phi,W)$ with $\phi > 0$. Note that a solution to
the limit equation with $\al_0 = 0$ would be a conformal Killing
field, but our standing assumption is that such do not exist on
$(M,g)$.

Using standard elliptic theory one can prove the following
Proposition, we leave out the details of the argument.
\begin{prop} \label{interior_data} The set of metrics $g$ and
  functions $\tau$ for which Equation \eqref{eqLimit1} has no solution
  is open with respect to the $C^1$-topology.
\end{prop}
We conclude that Theorem \ref{main} not only tells us that
\eqref{hamiltonian}-\eqref{momentum} have a solution if the limit
equation has no solution, but also that $g,\tau,\sigma$ are in the
interior of the set of data for which
\eqref{hamiltonian}-\eqref{momentum} have solutions.

Note that we do not make any claim about uniqueness. Since our
construction use the Schauder fixed point theorem as a central step we
cannot draw any conclusion about the uniqueness of solutions.

We continue with some applications of Theorem \ref{main}.

\begin{cor} \label{Cor1} Let data be given on $M$ as specified in
  \eqref{givendata} and assume that \eqref{assumptions} holds. Let
  $\lambda > 0$. Assume that the Ricci curvature $\Ric$ of $(M,g)$
  satisfies $\Ric \leq -\lambda g$ and that
  \begin{equation*}
    \left\| 
      \frac{d \tau}{\tau}
    \right\|_{ L^\infty} < \sqrt{\frac{n}{2(n-1)} \lambda}, 
  \end{equation*}      
  then the system of equations \eqref{hamiltonian}-\eqref{momentum}
  admits a solution $(\phi,W)$ with $\phi > 0$. Furthermore, the set
  of solutions $(\phi,W) \in W^{2, p}_+ \times W^{2,p}$ is compact.
\end{cor}

The next result is less explicit but holds without any curvature
assumptions. We define
\begin{equation*}
  C_g
  \definedas  
  \inf 
  \frac{{\left( \int_M |LV|^2 \, dv \right)}^{\frac{1}{2}}}
  {{\left( \int_ M |V|^N \, dv \right)}^{\frac{1}{N}}} ,
\end{equation*}
where the infimum is taken over all smooth $1$-forms $V$ on $M$ with
$V \not\equiv 0$. By standard elliptic estimates the constant $C_g$ is
positive if there are no conformal Killing vector fields on $(M,g)$,
see Lemma \ref{lem_Cg_positive} in the Appendix.

\begin{cor} \label{Cor2} Let data be given on $M$ as specified in
  \eqref{givendata} and assume that \eqref{assumptions} holds. If
  \begin{equation*}
    \left\| 
      \frac{d \tau}{\tau}
    \right\|_{L^n} < \frac{1}{2} \sqrt{\frac{n}{n-1}} \; C_g, 
  \end{equation*} 
  then the system of equations \eqref{hamiltonian}-\eqref{momentum}
  admits a solution $(\phi,W)$ with $\phi > 0$. Furthermore, the set
  of solutions $(\phi,W) \in W^{2, p}_+ \times W^{2,p}$ is compact.
\end{cor}

Since we do not know all the solutions of
\eqref{hamiltonian}-\eqref{momentum} a natural question is the
following: given $M$, $\tau$, and $\sigma$ as above, what can we say
about the set of metrics $g$ for which the system of equations
\eqref{hamiltonian}-\eqref{momentum} has a solution? Is it a large set
in some appropriate sense? The next result gives a partial answer to
this question.

\begin{cor} \label{Cor3} 
Let $\tau \in W^{1, p}$ be a function on $M$ which vanishes
nowhere. Let $\mathcal{R}(M,\tau)$ be the set of metrics on $M$ such
that for all symmetric, trace- and divergence-free $(0,2)$-tensors 
$\sigma \in W^{1,p}$ the assumptions \eqref{assumptions} hold and the
system of equations \eqref{hamiltonian}-\eqref{momentum} admits a
solution with $\phi > 0$. Then the set $\mathcal{R}(M,\tau)$ is dense
with respect to the $C^0$-topology in the set $\mathcal{R}(M)$ of all
metrics on $M$.
\end{cor}

Corollary \ref{Cor3} leads to the question if the set 
$\mathcal{R}(M,\tau,\sigma)$ is dense in a stronger topology. 
From a physical point of view,  the gravitational field depends on
the Levi-Civita connection and hence on the first derivatives of $g$.
So, a $C^1$-density result would be interesting. At the moment, we are
not able to obtain such a statement. One could speculate that the
limit equation \eqref{eqLimit1} never admits a non-trivial
solution. The following proposition says that this is not true.   

\begin{prop} \label{example_nd} 
On the sphere $S^n$ there exists a
metric $g$ and a function $\tau$ such that the
limit equation \eqref{eqLimit1} for $\tau$, $g$ has a non-trivial
solution for some $\al_0 \in (0,1]$.
\end{prop}

The proof is surprisingly complicated, it uses a contradiction
argument which might be unexpected for such an existence result.

Our last result is not related to the limit equation. It is a near-CMC
non-existence result which strengthens previous results of J. Isenberg
and N. {\'O}~Murchadha \cite[Theorem 2]{IsenbergOMurchadha}.

\begin{theorem} \label{nonexist} Let data be given on $M$ as specified
  in \eqref{givendata}. Assume that $\tau$ vanishes nowhere, $(M,g)$
  has no conformal Killing vector fields, and the operator
  \begin{equation*}
    \frac{4(n-1)}{n-2} \frac{N+1}{\left(\frac{N}{2}+1\right)^2} \Delta + R
    =
    \frac{3n-2}{n-1}\Delta + R 
  \end{equation*}
  is non-negative. Then if $\sigma \equiv 0$ and
  \begin{equation*}
    \left\| 
      \frac{d \tau}{\tau}
    \right\|_{L^n} <  \frac{1}{2} \sqrt{\frac{n}{n-1}} \; C_g, 
  \end{equation*} 
  there is no solution $(\phi,W) \in W^{2, p}_+ \times W^{2,p}$ of
  \eqref{hamiltonian}-\eqref{momentum}.
\end{theorem}
The assumption of non-negativity of this operator is weaker than $R
\geq 0$ but stronger than $Y(g) \geq 0$. This last fact is natural
since the conformal method is not conformally covariant. Also, the
assumptions made for Theorem \ref{main} and Theorem \ref{nonexist} are
mutually exclusive.

In Section \ref{sec_proof_main} we introduce the subcritical
perturbation of the conformal constraint equations and prove Theorem
\ref{main}. In Section \ref{sec_proof_cor} we give proofs of the
corollaries and of Proposition \ref{example_nd}. Section
\ref{sec_proof_nonexist} is devoted to a proof of the non-existence
result in Theorem \ref{nonexist}, and finally in the Appendix we prove
that the constant $C_g$ is positive.

The methods of this paper will be applied to the vacuum constraint
equations on asymptotically hyperbolic manifolds in
\cite{GicquaudSakovich}.  The asymptotically euclidean case will be
treated in a forthcoming paper.

We thank Jim Isenberg for showing us the idea of using subcritical
perturbations of the conformal constraint equations. Also we thank
Piotr Chru{\'s}ciel for pointing out a mistake in an earlier version
of this article. Further, we would like to thank Bernd Ammann, Lars
Andersson, Erwann Delay, and Olivier Druet for helpful discussions.

\section{Proof of Theorem \ref{main} } \label{sec_proof_main}

In all of this section we assume that data is given on $M$ as
specified in \eqref{givendata} and we assume that \eqref{assumptions}
holds. We will prove that the system
\eqref{hamiltonian}-\eqref{momentum} then admits a solution $(\phi,W)$
with $\phi > 0$ if the limit equation does not admit any non-trivial
solution. This proof proceeds in several steps.

\subsection{The subcritical system}

Let $0< \ep < 1$. We begin by introducing the subcritical system where
the exponent of $\phi$ in \eqref{momentum} is decreased by $\ep$,
\begin{subequations}
  \begin{align}
    \frac{4(n-1)}{n-2} \Delta \phi + R \phi &= -\frac{n-1}{n} \tau^2
    \phi^{N-1} + |\sigma + LW|^2 \phi^{-N-1},
    \label{hamiltonian_ep}   \\
    -\frac{1}{2} L^* L W &= \frac{n-1}{n} \phi^{N-\ep} d\tau.
    \label{momentum_ep}
  \end{align}
\end{subequations}
In the following Proposition we follow the method of Maxwell
\cite{MaxwellNonCMC} to show existence of solutions of these
subcritical equations.

\begin{prop} \label{existence_subcrit} Let data be given on $M$ as
  specified in \eqref{givendata} and assume that \eqref{assumptions}
  holds. Also, let $0 < \ep < 1$.  Then there exists at least one
  solution $(\phi, W)$ of the subcritical constraint equations
  \eqref{hamiltonian_ep}-\eqref{momentum_ep}.
\end{prop}

The proof of this proposition relies on a modified version of the
Schauder fixed point theorem, see for example \cite[Theorem 11.1, p.
279]{GilbargTrudinger}. Let $\phi \in L^\infty$, $\phi >0$. By
\cite[Proposition 5]{MaxwellNonCMC} there exists a unique $W \in W^{2,
  p}$ such that
\begin{equation}\label{momentum1}
  -\frac{1}{2} L^* L W = 
  \frac{n-1}{n} \phi^{N-\ep} d\tau,
\end{equation}
and by \cite[Proposition 2]{MaxwellNonCMC} there is a unique $\psi \in
W^{2, p}_+$ satisfying
\begin{equation}\label{hamiltonian1}
  \frac{4(n-1)}{n-2} \Delta \psi + R \psi = 
  -\frac{n-1}{n} \tau^2 \psi^{N-1} + |\sigma + LW|^2 \psi^{-N-1}.
\end{equation}
We define
\begin{equation*}
  N_\ep(\phi) \definedas \psi,
\end{equation*}
and prove the following result.

\begin{lemma} \label{supinf} There exists a constant $a_\ep>0$ such
  that for all $b \leq a_\ep$, there is a constant $K_b$ depending
  only on $b$, $g$, $\tau$, and $\sigma$, but not on $\ep$ such that
  \begin{equation*}
    K_b \leq N_\ep(\phi) \leq a_\ep
  \end{equation*}
  for any $\phi$ satisfying $0 < \phi \leq b$.  In addition if $Y(g) <
  0$, $K_b$ does not depend on $b$.
\end{lemma}

All the arguments for this proof can be found in \cite[Proposition
10]{MaxwellNonCMC}, even though a slightly weaker conclusion is
formulated there. To prove Proposition \ref{existence_subcrit} above
we only need the weaker version, later we will apply Lemma
\ref{supinf} as stated here.

\begin{proof} 
  Let $\phi$ be a positive function and let $W$ satisfy the perturbed
  vector equation \eqref{momentum1}. Assume that $\theta_-$ is a
  subsolution and $\theta_+$ is a supersolution of Equation
  \eqref{hamiltonian1} with $\theta_- \leq \theta_+$. Then it is well
  known that that there is a solution $\theta$ of Equation
  \eqref{hamiltonian1} which satisfies $\theta_- \leq \theta \leq
  \theta_+$, see for example \cite{Isenberg}. Since solutions of
  \eqref{hamiltonian1} are unique for a given $W$ we conclude that
  $\theta = N_\ep(\phi)$ and $\theta_- \leq N_\ep(\phi) \leq
  \theta_+$. To prove Lemma \ref{supinf}, we show that $\theta_+ =
  a_\ep$ is a supersolution if $a_\ep$ is a large real number and if
  $0 < b \leq a_\ep$ is such that $\phi \leq b$, then we can find a
  subsolution $\theta_-$ whose minimum $K_b$ depends only on $b$.

  By standard elliptic theory, there exists a constant $C_0$ depending
  only on $g$ such that any solution $W$ of Equation \eqref{momentum1}
  satisfies
  \begin{equation} \label{ellip_estim} \| L W \|_{L^\infty} \leq C_0
    \|d\tau\|_{L^p} \| \phi \|_{L^\infty}^{N-\ep}.
  \end{equation}
  Let $a_\ep > 0$ be large enough so that
  \begin{equation*} 
    Ra_\ep + \frac{n-1}{n}\tau^2 a_\ep^{N-1} 
    - 2 \|\sigma \|_{L^\infty}^2 a_\ep^{-N-1} 
    - 2 C_0^2  \|d\tau\|_{L^p}^2 a_\ep^{N - 1 - 2\ep} 
    \geq 0,
  \end{equation*} 
  such an $a_\ep$ can be found since we assume that $\tau$ vanishes
  nowhere and $\ep > 0$. At this step we crucially use that the
  equations have been made subcritical, here is also the only place we
  use this fact. We set $\theta_+ \definedas a_\ep$, this is then a
  supersolution of Equation \eqref{hamiltonian1}. We now find a
  subsolution, for this let $b \leq a_\ep$ and assume that $\phi \leq
  b$.

  We first study the case when $Y(g) \geq 0$. Then after an
  appropriate conformal change $\overline{g} \definedas e^{2u} g$ the
  scalar curvature $\overline{R}$ is non-negative. As in
  \cite[Proposition 10]{MaxwellNonCMC} we consider the solution $\eta$
  of
  \begin{equation} \label{equation_eta} \frac{4(n-1)}{n-2} \Delta \eta
    + \left( R + \frac{n-1}{n} \tau^2 \right) \eta = |\sigma + LW|^2.
  \end{equation} 
  Under the conformal change this equation is equivalent to the
  equation
  \begin{equation*}
    \frac{4(n-1)}{n-2} \overline{\Delta} \overline{\eta} 
    + \left(\overline{R} + \frac{n-1}{n} e^{-2u} \tau^2 \right)
    \overline{\eta} 
    =  
    e^{-(n+2)u/2} |\sigma + LW|^2
  \end{equation*}
  for $\overline{\eta} \definedas e^{-(n-2)u/2} \eta$. Since
  $\overline{R} + \frac{n-1}{n} e^{-2u} \tau^2$ is positive we know
  there exists a positive solution $\overline{\eta}$ of this equation,
  and thus there exists a positive solution $\eta$ of Equation
  \eqref{equation_eta}. We set $\eta_\al \definedas \al\eta$ for a
  positive constant $\al$, so that
  \begin{equation*}
    \begin{split}
      &\frac{4(n-1)}{n-2} \Delta \eta_\al + R \eta_\al + \frac{n-1}{n}
      \tau^2 \eta_\al^{N-1}
      - |\sigma + LW|^2 \eta_\al^{-N-1} \\
      &\quad = \frac{n-1}{n} \tau^2 \left( \al^{N-1} \eta^{N-1} - \al
        \eta \right) + |\sigma + LW|^2 \left(\al - \al^{-N-1}
        \eta^{-N-1} \right) .
    \end{split}
  \end{equation*}
  Choose $\alpha$ small enough so that this is non-positive, then
  $\theta_- \definedas \alpha \eta$ is a subsolution of Equation
  \eqref{hamiltonian1}. The choice of the $\alpha$ depends only on
  $\max(\eta)$, so the constant $K_b \definedas \min \th_-$ depends
  only on $\max(\eta)$ and $\min(\eta)$. By \cite[Proposition
  9]{MaxwellNonCMC} we have
  \begin{equation*}
    \max(\eta) 
    \leq 
    C_1 \left\| |\sigma + LW|^2 \right\|_{L^p}
    \leq
    C'_1 \left(\left\| \sigma \right\|_{L^{2p}}^{\frac{1}{2}}
      + \left\| LW \right\|_{L^{2p}}^{\frac{1}{2}}\right) ,
  \end{equation*}
  and
  \begin{equation*}
    \min(\eta) 
    \geq 
    C_2 \int_M |\sigma + LW|^2 \, dv 
    \geq
    C_2 \int_M |\sigma|^2 \, dv, 
  \end{equation*}
  where $C_1,C_2$ depend only on $g$ and $\tau$. Using
  \eqref{ellip_estim} we can further estimate
  \begin{equation*}
    \left\|LW\right\|_{L^{2p}}
    \leq 
    \vol(M)^{\frac{1}{2p}} \left\|LW\right\|_{L^\infty}
    \leq 
    \vol(M)^{\frac{1}{2p}} C_0 \|d\tau\|_{L^p} \| \phi \|_{L^\infty}^{N-\ep}
    \leq
    C'_0 \|d\tau\|_{L^p} b^{N-\ep},
  \end{equation*}
  where $C'_0 \definedas \vol(M)^{\frac{1}{2p}} C_0$, and we conclude
  that $K_b$ only depends on $b$, $g$, $\tau$, and $\sigma$. This
  finishes the proof in the first case.

  Now we study the case when $Y(g) < 0$. By \cite[Theorem 6.7, p.
  197]{Aubin} there exists a positive function $\eta$ such that
  $\widetilde{g} \definedas \eta^{N-2} g$ has scalar curvature
  $\widetilde{R} = - \frac{n-1}{n} \tau^2$. The function $\eta$ solves
  the equation
  \begin{equation*}
    \frac{4(n-1)}{n-2} \Delta \eta + R \eta 
    = 
    -\frac{n-1}{n} \tau^2 \eta^{N-1},
  \end{equation*} 
  and it follows that $\theta_- \definedas \eta$ is a subsolution of
  Equation \eqref{hamiltonian1}. Note that in this case, $K_b = \min
  \theta_-$ does not depend on $b$. This ends the proof of Lemma
  \ref{supinf}.
\end{proof}

The next step is to prove that the map $N_\ep$ is continuous. For this
we define maps
\begin{equation*}
  \Upsilon_\epsilon: L^\infty_+ \ni \phi \mapsto W \in W^{2, p}
\end{equation*}
where $W$ is the unique solution of the perturbed vector equation
\eqref{momentum1} given $\phi$, and
\begin{equation*}
  \Lambda_\epsilon: C^1 \ni  W \mapsto \psi \in L^\infty_+
\end{equation*}
where $\psi$ is the unique solution of the Lichnerowicz equation
\eqref{hamiltonian1} given $W$. Note that $N_\ep = \Lambda_\ep \circ
I_p \circ \Upsilon_\ep$ where $I_p : W^{2,p} \to C^1$ is the Sobolev
injection.

\begin{lemma} \label{lmContinuity} The maps $\Upsilon_\epsilon,
  \Lambda_\epsilon$, and thus also $N_\ep$, are continuous.
\end{lemma}

\begin{proof}
  The map $\Upsilon_\ep: \phi \mapsto W = (-\frac{1}{2} L^* L )^{-1}
  ((n-1) \phi^{N-\ep} d\tau )$ can be decomposed as
  \begin{equation*}
    \begin{array}{cccccccc}
      L^\infty_+ & \to & L^\infty_+ & \to & L^p & \to & W^{2, p} \\
      \phi & \mapsto & \phi^{N-\ep} & \mapsto & (n-1) \phi^{N-\ep} d\tau 
      & \mapsto 
      & \left(-\frac{1}{2} L^* L\right)^{-1} \left((n-1) \phi^{N-\ep}
        d\tau\right)
    \end{array}
  \end{equation*}
  and the continuity of each arrow here is straightforward.

  Next we prove that $\Lambda_\ep$ is continuous, we give a proof
  which is independent of the Yamabe class of $g$. Let $W_i \in C^1$,
  $i=0,1$, be arbitrary. Denote the corresponding solutions of
  Equation \eqref{hamiltonian1} by $\psi_i$ and set $u_i = \ln
  \psi_i$. Then $u_i$ satisfies the equation
  \begin{equation*}
    \frac{4(n-1)}{n-2} 
    \left( \Delta u_i + \left| d u_i \right|^2 \right) + R
    = 
    - \frac{n-1}{n} \tau^2 e^{(N-2) u_i} 
    + \left|\sigma+L W_1\right|^2 e^{-(N+2) u_i}.
  \end{equation*}
  Subtracting the equation for $u_0$ from the equation for $u_1$ we
  obtain
  \begin{equation*}
    \begin{split}
      &\frac{4(n-1)}{n-2} \left( \Delta (u_1 - u_0) + \langle d (u_1 +
        u_0), d(u_1 - u_0) \rangle
      \right) \\
      &\qquad =
      - \frac{n-1}{n} \tau^2 \left( e^{(N-2) u_1} - e^{(N-2) u_0} \right) \\
      &\qquad \qquad + \left|\sigma+L W_1\right|^2 e^{-(N+2) u_1} -
      \left|\sigma+L W_2\right|^2 e^{-(N+2) u_0}.
    \end{split}
  \end{equation*}
  Setting $u_\lambda \definedas (1 - \lambda) u_0 + \lambda u_1$ for
  $0 \leq \lambda \leq 1$ the previous equation can be rewritten as
  \begin{equation*}
    \begin{split}
      &\left( \left|\sigma + L W_1 \right|^2 - \left|\sigma + L W_0
        \right|^2
      \right) e^{-(N+2) u_0} \\
      &\qquad = \frac{4(n-1)}{n-2} \left( \Delta (u_1 - u_0) + \langle
        d (u_1 + u_0), d(u_1 - u_0) \rangle
      \right) \\
      &\qquad \qquad + \frac{1}{N-2} \frac{n-1}{n} \tau^2 (u_1 - u_0)
      \int_0^1 e^{(N-2) u_\lambda} d\lambda \\
      &\qquad \qquad + \frac{1}{N+2} \left|\sigma + L W_1\right|^2
      (u_1 - u_0) \int_0^1 e^{-(N+2) u_\lambda} d\lambda .
    \end{split}
  \end{equation*}
  Note that by Lemma \ref{supinf} we have $\psi_i \geq K_b$, so
  $e^{(N-2) u_\lambda} \geq K_b^{N-2}$, and
  \begin{equation*}
    \begin{split}
      &\frac{1}{N-2} \frac{n-1}{n} \tau^2 \int_0^1 e^{(N-2) u_\lambda}
      d\lambda + \frac{1}{N+2} \left|\sigma + L W_1\right|^2
      \int_0^1 e^{-(N+2) u_\lambda} d\lambda \\
      &\qquad \geq c_b \definedas \frac{1}{N-2} \frac{n-1}{n} \tau_0^2
      K_b^{N-2},
    \end{split}
  \end{equation*}
  where $\tau_0^2 \definedas \min_M \tau^2$ is positive by assumption.
  We are now in a position to apply the maximum principle
  \cite[Theorem 8.1, p. 179]{GilbargTrudinger} to the function $u_1 -
  u_0$ and we obtain
  \begin{equation*}
    \begin{split}
      \left\| u_1 - u_0 \right\|_{L^\infty} &\leq \frac{1}{c_b}
      \left\| \left( \left|\sigma + LW_1\right|^2 - \left|\sigma +
            LW_0\right|^2
        \right) e^{-(N+2) u_0} \right\|_{L^\infty} \\
      &\leq \frac{b^{-N-2}}{c_b} \left\| \left|\sigma + L W_1\right|^2
        - \left|\sigma + L W_0\right|^2
      \right\|_{L^\infty}\\
      &\leq \frac{b^{-N-2}}{c_b} \left\|2 \sigma + L (W_1 + W_0)
      \right\|_{L^\infty}
      \left\|L(W_1 - W_0) \right\|_{L^\infty}\\
      &\leq \frac{2 b^{-N-2}}{c_b} \left\|2 \sigma + L (W_1 + W_0)
      \right\|_{L^\infty} \left\|W_1 - W_0\right\|_{C^1}.
    \end{split}
  \end{equation*}
  This proves that the map $\Lambda_\ep : W \mapsto \psi$ is Lipschitz
  continuous on any subset $\{ W \in C^1 \mid \| W \|_{C^1} \leq C \}$
  for any $C$ large enough.
\end{proof}

We now prove the existence of solutions to the subcritical system.

\begin{proof}[Proof of Proposition \ref{existence_subcrit}]
  We follow the arguments of \cite[Section 4.2]{MaxwellNonCMC} and
  define $U \definedas \{ \phi \in L^{\infty} \mid K_{a_\ep} \leq \phi
  \leq a_\ep \}$.  Here $K_{a_\ep}$ is the constant given by Lemma
  \ref{supinf} applied with $b=a_\ep$. From Lemma \ref{supinf},
  $N_\ep$ maps $U$ into itself. The set $\Upsilon_\ep(U)$ is bounded.
  Because of the Rellich-Kondrakov theorem, $I_p \circ \Upsilon_\ep
  (U)$ is relatively compact. Hence $N_\ep(U) = \Lambda_\ep \circ I_p
  \circ \Upsilon_\ep (U) \subset U$ is also relatively compact. From
  the fact that $U$ is a convex closed subset of $L^\infty$, the
  closed convex hull $V$ of $N_\ep(U)$ is contained in $U$, hence $V$
  is a compact convex subset and the map $N_\ep$ maps $V$ into itself.
  By the Schauder fixed point theorem, there exists a fixed point for
  $N_\ep$ in $V$, namely a couple $(\phi_\ep,W_\ep)$, $\phi_\ep>0$,
  solving the system \eqref{hamiltonian_ep}-\eqref{momentum_ep}. By
  standard regularity theorems we get $\phi_\ep \in W^{2, p}_+$ and
  $W_\ep \in W^{2, p}$.
\end{proof}

\subsection{Convergence of subcritical solutions}

\def\phitil{\widetilde{\phi}} \def\Wtil{\widetilde{W}}
\def\sigmatil{\widetilde{\sigma}}

Let $\epsilon \in [0, 1)$ be arbitrary and let $(\phi, W)$ be a
solution of the subcritical equations
\eqref{hamiltonian_ep}-\eqref{momentum_ep}. Define the {\em energy} of
this solution by
\begin{equation} \label{defEnergy} \gamma(\ph, W) \definedas \int_M
  |LW|^2 \, dv,
\end{equation}
and set $\gammatil \definedas \max \{\gamma, 1\}$. We first prove that
$\gammatil$ controls the $L^\infty$-norm of $\phi$. Note that
$\epsilon = 0$ is allowed, so the result applies also to solutions of
\eqref{hamiltonian}-\eqref{momentum}.

\begin{prop} \label{propBoundPhi} There exists a constant $C$ such
  that for any $\epsilon \in [0, 1)$ and any pair $(\phi, W)$ solving
  the subcritical equations \eqref{hamiltonian_ep}-\eqref{momentum_ep}
  it holds that
  \begin{equation*}
    \phi \leq C \gammatil(\ph, W)^{\frac{1}{2N}}.
  \end{equation*}
\end{prop}

\begin{proof} 
  To abbreviate we set $\gammatil \definedas \gammatil(\ph, W)$. We
  rescale $\ph$, $W$ and $\sigma$ as follows,
  \begin{equation*}
    \phitil \definedas \gammatil^{-\frac{1}{2N}} \ph, \quad
    \Wtil \definedas \gammatil^{-\frac{1}{2}} W, \quad
    \sigmatil \definedas \gammatil^{-\frac{1}{2}} \sigma.
  \end{equation*}
  The subcritical equations \eqref{hamiltonian_ep}-\eqref{momentum_ep}
  can then be rewritten as
  \begin{subequations}
    \begin{align}
      \label{hamiltonian_ep_rescaled}
      \frac{1}{\gammatil^{\frac{1}{n}}} \left(\frac{4(n-1)}{n-2}
        \Delta \phitil + R \phitil \right) + \frac{n-1}{n} \tau^2
      \phitil^{N-1} &=
      \left| \sigmatil + L\Wtil \right|^2 \phitil^{-N-1} , \\
      \label{momentum_ep_rescaled}
      -\frac{1}{2} L^*L \Wtil &= \frac{n-1}{n}
      \gammatil^{-\frac{\epsilon}{2N}} \phitil^{N-\epsilon} d\tau.
    \end{align}
  \end{subequations}
  Due to the rescaling we have
  \begin{equation*}
    \int_M \left|L \Wtil\right|^2 \, dv \leq 1.
  \end{equation*}
  The proof proceeds in three steps.

\begin{step} 
  Bound for the $L^2$-norm of $\phitil^N$ \label{stBound1}
\end{step}

Multiplying equation \eqref{hamiltonian_ep_rescaled} by
$\phitil^{N+1}$ and integrating over $M$, we obtain
\begin{equation*}
  \begin{split}
    &\frac{1}{\gammatil^{\frac{1}{n}}} \int_M \left(
      \frac{4(n-1)}{n-2} \phitil^{N+1} \Delta \phitil + R
      \phitil^{N+2} \right) \, dv
    + \frac{n-1}{n} \int_M \tau^2 \phitil^{2N} \, dv \\
    &\qquad =
    \int_M \left|\sigmatil + L\Wtil\right|^2 \, dv \\
    &\qquad = \int_M \left|\sigmatil\right|^2 \, dv
    + \int_M \left|L\Wtil\right|^2 \, dv \\
    &\qquad \leq \int_M \left|\sigmatil\right|^2 \, dv + 1.
  \end{split}
\end{equation*}
Integration by parts tells us that the first term in the first
integral is non-negative. We set $\tau_0^2 \definedas \min_M \tau^2$,
which is positive by assumption, and conclude that
\begin{equation*}
  \frac{1}{\gammatil^{\frac{1}{n}}}\int_M R \phitil^{N+2} \, dv
  + 
  \frac{n-1}{n} \tau_0^2\int_M \phitil^{2N} \, dv 
  \leq
  \int_M \left|\sigmatil\right|^2 \, dv + 1. 
\end{equation*}
Here the first term can be estimated using the H\"older inequality
\begin{equation*}
  \left|\int_M R \phitil^{N+2} \, dv \right| 
  \leq 
  \left( \int_M |R|^{\frac{2N}{N-2}} \, dv \right)^{\frac{N-2}{2N}} 
  \left( \int_M \phitil^{2N} \, dv \right)^{\frac{N+2}{2N}} .
\end{equation*}
Remark that $\frac{2N}{N-2} = n < p$, hence, setting
\begin{equation*}
\lambda = \left|\int_M \left|R\right|^n \right|^\frac{1}{n} < \infty,
\end{equation*}
we obtain:
\begin{equation*}
  -\frac{\lambda}{\gammatil^{\frac{1}{n}}} 
  \left( \int_M \phitil^{2N} \, dv \right)^{\frac{N+2}{2N}} 
  + 
  \frac{n-1}{n} \tau_0^2\int_M \phitil^{2N} \, dv 
  \leq 
  \int_M \left|\sigmatil\right|^2 \, dv + 1.
\end{equation*}
Since $\frac{N+2}{2N} < 1$ and $\gammatil \geq 1$ by definition, we
obtain that the $L^2$-norm of $\phitil^N$ is bounded independently of
the choice of $\epsilon$, $\ph$, and $W$.

\begin{step} 
  Induction step
\end{step}

We shall now show that the previous $L^2$-estimate for $\ph^N$ can be
iteratively improved to give $L^p$-estimates for $p \geq 2$. Set $p_0
= 2$, we construct inductively an increasing sequence $p_i$ such that
$\phitil^N$ is bounded in the $L^{p_i}$-norm. Assume by induction that
$\phitil^N$ is bounded in $L^{p_i}$. We first improve our estimate for
$\Wtil$. From Equation \eqref{momentum_ep_rescaled} we get
\begin{equation*}
  \frac{1}{2} \left\|L^*L \Wtil\right\|_{L^{q_i}} 
  \leq 
  \frac{n-1}{n}
  \left\| \phitil^{N - \epsilon}\right\|_{L^{p_i}}
  \left\| d\tau \right\|_{L^p},
\end{equation*}
where $q_i$ is such that $\frac{1}{q_i}=\frac{1}{p_i}+\frac{1}{p}$.
Note that the sequence $q_i$ is increasing since $p_i$ is and that
$\frac{1}{q_0} = \frac{1}{2} + \frac{1}{p} < 1$ since $p > n \geq 3$.
{}From Young's inequality, we infer
\begin{equation*}
  \phitil^{N-\epsilon} 
  \leq 
  \frac{N - \epsilon}{N} \phitil^N + \frac{\epsilon}{N},
\end{equation*}
hence
\begin{equation*}
  \left\|\phitil^{N-\epsilon}\right\|_{p_i}
  \leq 
  \frac{N-\epsilon}{N} \left\|\phitil^N\right\|_{p_i} 
  + \frac{\epsilon}{N} \vol(M)^{\frac{1}{p_i}} 
  \leq \left\|\phitil^N\right\|_{p_i}
  + \frac{1}{N} \max\left\{1, \vol(M)\right\}. 
\end{equation*}
{}From standard elliptic estimates and the fact that $(M,g)$ admits no
conformal Killing vector field, there exists a constant $C_i$ such
that  
\begin{equation*}
  \left\|\Wtil\right\|_{W^{2, q_i}} 
  \leq 
  C_i 
  \left(\left\| \phitil^N \right\|_{p_i} 
    + \frac{1}{N} \max \left\{1, \vol(M) \right\} \right)
  \left\| d\tau \right\|_{L^p}. 
\end{equation*}
Assume that $q_i < n$. By the Sobolev injection, we deduce that
$L\Wtil$ is bounded in $L^{r_i}$ where $r_i$ is such that
$\frac{1}{r_i} = \frac{1}{q_i} - \frac{1}{n}$. We now estimate
$\phitil^N$ by means of the rescaled Lichnerowicz equation
\eqref{hamiltonian_ep_rescaled}. The method is similar to Step
\ref{stBound1}. We multiply \eqref{hamiltonian_ep_rescaled} by
$\phitil^{N+1+N k_i}$ where $k_i$ satisfies $\frac{2}{r_i} +
\frac{k_i}{p_i} = 1$ and integrate over $M$,
\begin{equation*}
  \begin{split}
    &\int_M \left| \sigmatil + L\Wtil \right|^2 \phitil^{N k_i} \, dv \\
    &\qquad = \frac{1}{\gammatil^{\frac{1}{n}}}\int_M \left(
      \frac{4(n-1)}{n-2} \phitil^{N+1+N k_i} \Delta \phitil
      + R \phitil^{N+2+N k_i}\right) \, dv \\
    &\qquad \qquad
    + \frac{n-1}{n} \int_M \tau^2 \phitil^{2N+N k_i} \, dv, \\
  \end{split}
\end{equation*}
{}from which we find
\begin{equation*}
  \begin{split}
    &2 \int_M \left(\left|\sigmatil\right|^2
      + \left|L\Wtil\right|^2\right) \phitil^{N k_i} \, dv \\
    &\qquad \geq \frac{1}{\gammatil^{\frac{1}{n}}} \int_M \left(
      \frac{4(n-1)}{n-2} \frac{N+1+N k_i}{\left(\frac{N + N
            k_i}{2}+1\right)^2} \left| d\phitil^{\frac{N+N k_i}{2}+1}
      \right|^2
      + R \phitil^{N+2+N k_i}\right) \, dv \\
    &\qquad \qquad
    + \frac{n-1}{n} \tau_0^2 \int_M \phitil^{2N+N k_i} \, dv \\
    &\qquad \geq \frac{1}{\gammatil^{\frac{1}{n}}} \int_M R
    \phitil^{N+2+N k_i} \, dv +
    \frac{n-1}{n} \tau_0^2 \int_M \phitil^{2N+N k_i} \, dv. \\
  \end{split}
\end{equation*}
By a method similar to the one we used in Step \ref{stBound1}, we
obtain that
\begin{equation*}
  \frac{1}{\gammatil^{\frac{1}{n}}}\int_M R \phitil^{N+2+N k_i} \, dv 
  \geq 
  - \frac{C_i}{\gammatil^{\frac{1}{n}}} 
  \left(\int_M \phitil^{(2 + k_i)N} \, dv \right)^{\frac{2+(1+k_i)N}{(2+k_i)N}}
\end{equation*}
for some constant $C_i$ that depends only on $g$ and $k_i$. We get the
following inequality
\begin{equation} \label{eqMainInequality}
  \begin{split}
    &- \frac{C_i}{\gammatil^{\frac{1}{n}}} \left(\int_M \phitil^{(2 +
        k_i)N} \, dv \right)^{\frac{2+(1+k_i) N}{(2+k_i)N}}
    + \frac{n-1}{n} \tau_0^2 \int_M \phitil^{2N+N k_i} \, dv \\
    &\qquad \leq 2 \int_M \left( \left| \sigmatil \right|^2 + \left|
        L\Wtil \right|^2 \right) \phitil^{N k_i} \, dv .
  \end{split}
\end{equation}
Here we have that $\left|\sigmatil\right|^2 + \left|\Wtil\right|^2$ is
bounded in $L^{\frac{r_i}{2}}$ and $\phitil^{N k_i}$ is bounded in
$L^{\frac{p_i}{k_i}}$, by the Young inequality we conclude that the
right hand side of \eqref{eqMainInequality} is bounded. From the fact
that $N > 2$, we have
\begin{equation*}
  \frac{2+(1+k_i)N}{(2+k_i)N} < 1.
\end{equation*} 
Hence by Inequality \eqref{eqMainInequality}, we conclude that
$\phitil^N$ is bounded in $L^{p_{i+1}}$ where $p_{i+1} = 2+k_i$.  A
simple calculation yields
\begin{equation*}
  \frac{p_{i+1}}{p_i} 
  = 
  1 + 2 \left(\frac{1}{n} - \frac{1}{p}\right) > 1.
\end{equation*}
Since
\begin{equation*}
  \frac{1}{q_i} 
  = 
  \frac{1}{p_i} + \frac{1}{p} < \frac{1}{p_i} + \frac{1}{n}
\end{equation*}
there exists an $i_0 > 0$ such that $q_{i_0} \geq n$ and $q_{i_0 - 1}
< n$. When $q_i > n$, the Sobolev space $W^{1, q_i}$ embeds into some
H\"older space so the strategy changes. This is the content of the
next step.

\begin{step}
  $L^\infty$ bound
\end{step}

Assume first that $q_{i_0} > n$. Since $\Wtil$ is bounded in
$W^{2,q_{i_0}}$ we know that $L\Wtil$ is bounded in the
$C^\alpha$-norm where $\alpha \definedas 1 - \frac{n}{q_{i_0}}$.
{}From the fact that the Laplacian acting on functions only involves
first order derivatives of the metric, it can be easily seen that the
function $\phitil$ is in $C^{2, \alpha}$. We can thus apply the
classical maximum principle to the equation
\eqref{hamiltonian_ep_rescaled}.  Let $x \in M$ be such that $\phitil$
attains its maximum value at $x$. At such a point, the following
inequality holds
\begin{equation*}
  \frac{1}{\gammatil^{\frac{1}{n}}} R \phitil 
  + \frac{n-1}{n} \tau^2 \phitil^{N-1}  
  \leq 
  \left|\sigmatil + L\Wtil\right|^2 \phitil^{-N-1},
\end{equation*}
or
\begin{equation*}
  \frac{1}{\gammatil^{\frac{1}{n}}} R \phitil^{N+2} 
  + \frac{n-1}{n} \tau^2 \phitil^{2 N} 
  \leq 
  \left|\sigmatil + L\Wtil\right|^2.
\end{equation*}
Hence $\phitil$ is bounded. The case $q_{i_0} = n$ can be avoided by
slightly decreasing the value of $p$.

We have proved that $\phitil$ is bounded by some constant which
depends only on $g$, $\tau$, and $\sigma$. Since $\ph =
\gammatil^{\frac{1}{2N}} \phitil$ this proves Proposition
\ref{propBoundPhi}.
\end{proof}

We can now study what happens when $\epsilon \to 0$. This is the
content of the next two lemmas. We will see that the behavior is
determined by the fact that the corresponding energies $\gamma$ are
bounded or not.
\begin{lemma} \label{lmBoundedEnergy} Assume that there exists
  sequences $\epsilon_i$ and $(\ph_i, W_i)$ such that $\epsilon_i \geq
  0$, $\epsilon_i \to 0$, and $(\ph_i, W_i)$ is a solution of the
  subcritical equations \eqref{hamiltonian_ep}-\eqref{momentum_ep}
  with $\epsilon = \epsilon_i$. Furthermore assume $\gamma(\ph_i,
  W_i)$ is bounded. Then there exists a subsequence of the $(\ph_i,
  W_i)$ which converges in the $W^{2, p}$-norm to a solution
  $(\ph_{\infty}, W_{\infty})$ of the conformal constraint equations
  \eqref{hamiltonian}-\eqref{momentum}.
\end{lemma}

\begin{proof} {}From the previous proposition, we know that the
  functions $\ph_i$ are uniformly bounded for the $L^\infty$-norm.
  {}From Equation \eqref{momentum_ep}, the sequence $W_i$ is uniformly
  bounded in $W^{2,p}$. Since $p > n$ we conclude by the
  Rellich-Kondrakov theorem that the map $L : W^{2, p} \to L^\infty$
  is compact. Up to extracting a subsequence, we can assume that the
  sequence $L W_i$ converges for the $L^\infty$-norm to some $L
  W_{\infty}$. Hence, from Lemma \ref{lmContinuity}, the functions
  $\ph_i$ converge in the $W^{2,p}$-norm (and therefore in the
  $L^\infty$-norm) to some $\ph_{\infty}$. Equation
  \eqref{momentum_ep} finally implies that the sequence $W_i$
  converges to $W_{\infty}$ in the $W^{2,p}$-norm.
\end{proof}

\begin{lem} \label{lmUnboundedEnergy} Assume that there exists
  sequences $\epsilon_i$ and $(\ph_i, W_i)$ such that $\epsilon_i \geq
  0$, $\epsilon_i \to 0$, and $(\ph_i, W_i)$ is a solution of the
  subcritical equations \eqref{hamiltonian_ep}-\eqref{momentum_ep}
  with $\epsilon = \epsilon_i$. Furthermore assume $\gamma(\ph_i, W_i)
  \to \infty$. Then there exists a non-zero solution $W \in W^{2, p}$
  of the limit equation
  \begin{equation*}
    -\frac{1}{2} L^*L W 
    = \al_0 \sqrt{\frac{n-1}{n}} |LW| \frac{d\tau}{\tau}.
  \end{equation*}
  for some $\al_0 \in (0,1]$.
\end{lem}

Note that, heuristically, when $\al_0= 1$, this equation is obtained
from the constraint equations \eqref{hamiltonian}-\eqref{momentum} by
discarding the terms $\frac{4(n-1)}{n-2} \Delta \ph + R \ph$ and
$\sigma$ in the Lichnerowicz equation \eqref{hamiltonian}. These are
the terms that disappear in the limit $\gammatil \to \infty$, see
Equation \eqref{hamiltonian_ep_rescaled}. The constant $\al_0$ will
appear as a defect of scale invariance when dealing with subcritical
systems, that is when $\ep_i \to 0$ but $\ep_i \not=0$. More precisely,
$\al_0$ will be defined as $\al_0 \definedas \lim \gamma(\ph_i,
W_i)^{-\frac{\epsilon_i}{2N}}$.

\begin{proof}
  Arguing as in the previous lemma, the $1$-forms $\Wtil_i$ are
  uniformly bounded in $W^{2, p}$. Without loss of generality we can
  assume that $\gamma(\ph_i, W_i) > 1$ for all $i$, so that
  \begin{equation*}
    \int_M \left|L\Wtil_i\right|^2 \, dv = 1.
  \end{equation*}
  Up to extracting a subsequence, we can then assume that $\Wtil_i$
  converges in the $C^1$-norm to some $\Wtil_{\infty}$. Note that
  \begin{equation*}
    \int_M \left|L\Wtil_{\infty}\right|^2 \, dv = 1
  \end{equation*} 
  so $\Wtil_{\infty} \not\equiv 0$. All we have to show now is that
  the functions $\phitil_i$ converge in the $L^\infty$-norm to
  $\phitil_{\infty} \in L^\infty$, where $\phitil_{\infty}$ is given
  by
  \begin{equation*}
    \phitil_{\infty}^N 
    =
    \sqrt{\frac{n}{n-1}} \tau^{-1} \left|L \Wtil_{\infty}\right|,
  \end{equation*}
  in other words, $\phitil_{\infty}$ is such that
  \begin{equation*}
    \frac{n-1}{n} \tau^2 \phitil_{\infty}^{N-1} 
    = 
    \left|L\Wtil_{\infty}\right| \phitil_{\infty}^{-N-1} .
  \end{equation*}
  If such a statement is proven, passing to the limit in
  \eqref{momentum_ep_rescaled}, we see that $\Wtil_\infty$ is a
  solution of the limit equation with $\al_0 \definedas \lim
  \gamma(\ph_i, W_i)^{-\frac{\epsilon_i}{2N}}$ which belongs to
  $[0,1]$ since $\gamma(\ph_i, W_i)$ is assumed to go to infinity.
  By assumption $(M,g)$ has no conformal Killing fields, which
  excludes the possibility that $\al_0 = 0$. Choose $\epsilon > 0$, we
  will show that there exists $i_0$ such that
  \begin{equation*}
    \left|\phitil_i - \phitil_{\infty}\right| < \epsilon.
  \end{equation*}
  for all $i \geq i_0$. For this, take two arbitrary $C^2$ functions
  $\phitil_\pm$ such that
  \begin{equation*}
    \begin{aligned}
      \phitil_{\infty} - \epsilon &\leq \phitil_-
      \leq \phitil_{\infty} - \frac{\epsilon}{2}, \\
      \phitil_{\infty} + \frac{\epsilon}{2} &\leq \phitil_+
      \leq \phitil_{\infty} + \epsilon.\\
    \end{aligned}
  \end{equation*}
  We first show that $\phitil_+$ is a supersolution of the rescaled
  Lichnerowicz equation \eqref{hamiltonian_ep_rescaled} if $i$ is
  large enough. By the maximum principle, we conclude that
  \begin{equation*}
    \phitil_i \leq \phitil_+ \leq \phitil_{\infty} + \epsilon.
  \end{equation*}
  Note first that $\phitil_+ \geq \frac{\epsilon}{2} > 0$. Multiplying
  the rescaled Lichnerowicz equation \eqref{hamiltonian_ep_rescaled}
  by $\phitil_+^{N+1}$, we have to show that
  \begin{equation*}
    \frac{\phitil_+^{N+1}}{\gammatil^{\frac{1}{n}}} 
    \left( \frac{4(n-1)}{n-2} \Delta \phitil_+ + R \phitil_+ \right) 
    + 
    \frac{n-1}{n} \tau^2 \phitil_+^{2N} 
    \geq 
    \left| \sigmatil + L\Wtil_i \right|^2.
  \end{equation*}
  Since
  \begin{equation*}
    \phitil_+^{2N} 
    \geq 
    \left(\phitil_{\infty} + \frac{\epsilon}{2}\right)^{2N} 
    \geq
    \phitil_{\infty}^{2N} + \left(\frac{\epsilon}{2}\right)^{2N}
  \end{equation*}
  the previous inequality will be satisfied provided that
  \begin{equation*}
    \frac{\phitil_+^{N+1}}{\gammatil^{\frac{1}{n}}}
    \left(\frac{4(n-1)}{n-2} \Delta \phitil_+ + R \phitil_+\right) 
    + 
    \left(\frac{\epsilon}{2}\right)^{2N} 
    \geq 
    \left| \sigmatil + L\Wtil_i \right|^2 
    - 
    \left|L\Wtil_{\infty}\right|^2.
  \end{equation*}
  Now note that both the terms
  \begin{equation*}
    \frac{\phitil_+^{N+1}}{\gammatil^{\frac{1}{n}}}
    \left(\frac{4(n-1)}{n-2} \Delta \phitil_+ + R \phitil_+\right)
  \end{equation*}
  and
  \begin{equation*}
    \left|\sigmatil + L\Wtil_i\right|^2 - \left|L\Wtil_{\infty}\right|^2
  \end{equation*} 
  tend uniformly to zero due to the facts that $\phitil_+ \in C^2$,
  $\gammatil_i \to \infty$ and $\Wtil_i \to \Wtil_{\infty}$. This
  proves that there exists $i_0$ such that for all $i \geq i_0$,
  $\phitil_+$ is a supersolution.

  We now prove that $\phitil_- \leq \phitil_i$ for $i$ large enough.
  Note first that $\phitil_i \geq \gamma^{-\frac{1}{2N}} K_b$ so this
  inequality will be fulfilled if we show that $\phitil_-$ is a
  subsolution of the Lichnerowicz equation everywhere it is positive
  since we can then apply the maximum principle on the set $\{ x \in M
  \mid \phitil_-(x) \geq \gamma^{-\frac{1}{2N}} K_b\}$. The proof is
  then exactly the same as for $\phitil_+$, except that we use the
  inequality
  \begin{equation*}
    \phitil_-^{2N} 
    \leq 
    \left(\phitil_{\infty} - \frac{\epsilon}{2}\right)^{2N} 
    \leq 
    \phitil_{\infty}^{2N} - \left(\frac{\epsilon}{2}\right)^{2N}.
  \end{equation*}
  We conclude that $\phitil_i$ converges to $\phitil_{\infty}$, which
  ends the proof of Lemma \ref{lmUnboundedEnergy}.
\end{proof}

With the results obtained so far we can now prove Theorem \ref{main}.

\begin{proof}[Proof of Theorem \ref{main}]
  Assume that the limit equation \eqref{eqLimit1} admits no non-zero
  solution for any $\al_0 \in (0,1]$. We first show that there exists a
  non-zero solution of the constraint equations. Indeed, for positive
  integers $i$ set $\epsilon_i \definedas 1/i$. From Lemma
  \ref{existence_subcrit}, there exists a solution $(\ph_i, W_i) \in
  W^{2,p}_+ \times W^{2,p}$ of the subcritical constraint equations
  \eqref{hamiltonian_ep}-\eqref{momentum_ep} where $\epsilon =
  \epsilon_i$. If the sequence $\gamma(\ph_i, W_i)$ was unbounded,
  there would exist a non-zero solution of the limit equation
  \eqref{eqLimit1} by Lemma \ref{lmUnboundedEnergy}. This contradicts
  the assumptions of the proposition. Hence, the sequence
  $\gamma(\ph_i, W_i)$ is bounded and Lemma \ref{lmBoundedEnergy}
  asserts that there exists a solution $(\ph_{\infty}, W_{\infty}) \in
  W^{2,p}_+ \times W^{2,p}$ of the constraint equations.

  To prove compactness, let $(\ph_i, W_i)$ be an arbitrary sequence of
  solutions of the conformal constraint equations. Applying Lemma
  \ref{lmUnboundedEnergy} with $\epsilon_i = 0$ we obtain that the
  sequence $\gamma(\ph_i, W_i)$ is bounded. Lemma
  \ref{lmBoundedEnergy} then asserts that a subsequence of $(\ph_i,
  W_i)$ converges. This ends the proof of Theorem \ref{main}.
\end{proof}

\begin{remark} \label{epi=0} We see from the proof of Lemma
  \ref{lmUnboundedEnergy} that $\al_0 = 1$ if the $\ep_i$ are all
  equal to $0$. This implies in particular if the limit equation
  \eqref{eqLimit1} has no non-trivial solution for $\al_0 = 1$, then
  the compactness in the statement of Theorem \ref{main} holds.
\end{remark}

\section{The limit equation: existence and non-existence results}
\label{sec_proof_cor}

In this section, we study solutions of the limit equation
\eqref{eqLimit1}. We first give non-existence results and prove
Corollaries \ref{Cor1}, \ref{Cor2}, and \ref{Cor3}. Then, we derive
some properties leading to the existence result in Proposition
\ref{example_nd}.

\subsection{Non-existence results}

We first prove the Limit equation has no solutions when the Ricci
curvature has a negative upper bound.

\begin{proof}[Proof of Corollary \ref{Cor1}]
  Under the assumption of Corollary \ref{Cor1}, we show that the limit
  equation \eqref{eqLimit1} does not have any non-trivial solution,
  Theorem \ref{main} then implies the existence of solutions to the
  conformal constraint equations. We proceed by contradiction and
  assume that $W \not\equiv 0$ is a solution of
  \begin{equation*}
    -\frac{1}{2} L^*L W 
    = \al_0 \sqrt{\frac{n-1}{n}} |LW| \frac{d\tau}{\tau} .
  \end{equation*}
  for some $\al_0 \in (0,1]$. Taking the scalar product with $W$,
  integrating over $M$, and using the H\"older inequality, we get
  \begin{equation} \label{eqlw2} \frac{1}{2} \int_M |LW|^2 \, dv \leq
    \sqrt{\frac{n-1}{n}} \left\| \frac{d \tau}{\tau}
    \right\|_{L^{\infty}} {\left( \int_M |LW|^2 \, dv
      \right)}^{\frac12} {\left( \int_M |W|^2 \, dv
      \right)}^{\frac12}.
  \end{equation} 
  The conformal Killing operator satisfies the Bochner type formula
  \begin{equation} \label{LW2} \frac{1}{2} \int_M \left|LV\right|^2 \,
    dv = \int_M \left( \left|\nabla V\right|^2 + \left(1 -
        \frac{2}{n}\right) \left|\mathrm{div} V\right|^2 - \Ric(V,V)
    \right) \, dv,
  \end{equation}
  see for example \cite[Chapter 2, \S 6]{Yano}. Since we assume that
  the Ricci curvature satisfies $\Ric \leq - \lambda g$ for some
  $\lambda > 0$, we get
  \begin{equation*}
    \int_M |LW|^2 \, dv \geq 2 \lambda \int_M |W|^2 \, dv.
  \end{equation*}
  Together with \eqref{eqlw2}, we obtain
  \begin{equation*}
    \left\| \frac{d \tau}{\tau} \right\|_{L^{\infty}} 
    \geq 
    \sqrt{\frac{n}{2(n-1)} \lambda}
  \end{equation*}
  which contradicts the assumption of Corollary \ref{Cor1}.
\end{proof}

The next result is that the limit equation has no solution under a
near CMC condition.
\begin{proof}[Proof of Corollary \ref{Cor2}]
  As in the proof of Corollary \ref{Cor1}, we assume that $W$ is a
  solution of Equation \eqref{eqLimit1}.  Taking the scalar product of
  this equation with $W$, integrating over $M$, and using the H\"older
  inequality, we find
  \begin{equation*}
    \frac{1}{2} \int_M |LW|^2 \, dv  \leq  \sqrt{\frac{n-1}{n}} \left\| 
      \frac{d \tau}{\tau}
    \right\|_{ L^n} 
    {\left( \int_M |LW|^2 \, dv \right)}^{\frac12} {\left( \int_M
        |W|^N  \, dv
      \right)}^{\frac{1}{N}}.
  \end{equation*}
  The definition of $C_g$ immediately implies that
  \begin{equation*}
    \left\| 
      \frac{d \tau}{\tau}
    \right\|_{ L^n} \geq  \frac{1}{2} \sqrt{\frac{n}{n-1}} \; C_g, 
  \end{equation*} 
  which contradicts the assumption of Corollary \ref{Cor2}.
\end{proof}

Finally we show that metrics with no solution to the limit equation
are dense in the $C^0$-topology.
\begin{proof}[Proof of Corollary \ref{Cor3}]
  Fix any metric $g$ and set
  \begin{equation*}
    c \definedas
    \left\| \frac{d \tau}{\tau} \right\|_{L^{\infty}(M,g)}^2.
  \end{equation*} 
  By a result of J. Lohkamp \cite[Theorem B]{Lohkamp} the set of
  metrics with a given upper bound on the Ricci curvature is dense in
  the space of all metrics with respect to the $C^0$-topology. Thus we
  can take a sequence of metrics $g_i$ tending to $g$ in $C^0$ with
  the property that
  \begin{equation*}
    \Ric^{g_i} \leq -\lambda g_i
  \end{equation*} 
  for all $i$, where $\lambda \definedas \frac{4 (n-1)}{n}c$.  For $i$
  large enough we have
  \begin{equation*}
    \left\| \frac{d \tau}{\tau} \right\|_{L^{\infty}(M,g_i)}^2 
    \leq 
    \frac{3}{2}c
    < 2 c
    = \frac{n}{2 (n-1)} \lambda.
  \end{equation*}
  By Corollary \ref{Cor1} we obtain the limit equation
  \eqref{eqLimit1} has no non-trivial solutions for all metrics $g_i$
  with $i$ large enough and hence, that the system of equations
  \eqref{hamiltonian}-\eqref{momentum} admits a solution for all
  metrics $g_i$ with $i$ large enough. This proves Corollary
  \ref{Cor3}.
\end{proof}

\subsection{Existence result for the limit equation}

For the proof of Proposition \ref{example_nd} we will use the
following result.

\begin{prop} \label{limit_equations} Let $M$ be a compact $n$-manifold
  equipped with a smooth Riemannian metric $g$. Let also $\tau > 0$ be
  smooth. Then at least one of the following statements is true.
  \begin{itemize}
  \item The limit equation \eqref{eqLimit1} has a non-trivial solution
    for some $\al_0 \in (0,1]$.
  \item For every trace-free and divergence-free smooth $(0,2)$-tensor
    $\sigma \not\equiv 0$ there exists a solution $W \in W^{2,p}$ of
    the equation
    \begin{equation} \label{eqLimit_sigma} -\frac{1}{2} L^*L W =
      \sqrt{\frac{n-1}{n}} |\sigma + LW| \frac{d\tau}{\tau} .
    \end{equation} 
  \end{itemize}
\end{prop}

The assumptions on the regularity of $g$, $\tau$, $\sigma$ in
Proposition \ref{limit_equations} could be weakened. To simplify the
proof we choose to assume smoothness of the data, which is enough for
the example we want to find.

\begin{proof}
  Assume that the limit equation \eqref{eqLimit1} has no solution and
  take a trace-free and divergence-free $(0,2)$-tensor $\sigma
  \not\equiv 0$. We need to prove that \eqref{eqLimit_sigma} then has
  a solution. Theorem \ref{main} implies that for all $a>0$ the
  conformal constraint equations for $(g,\tau, a \sigma)$ admit a
  solution. So, there exists $(\phi_a,W_a)$ with $\phi_a >0$ such that
  \begin{subequations}
    \begin{align}
      \frac{4(n-1)}{n-2} \Delta \phi_a + R \phi_a &= -\frac{n-1}{n}
      \tau^2 \phi_a^{N-1} + |a \sigma + LW_a|^2 \phi_a^{-N-1},
      \label{hamiltoniana}   \\
      -\frac{1}{2} L^* L W_a &= \frac{n-1}{n} \phi_a^N d\tau.
      \label{momentuma}
    \end{align}
  \end{subequations}
  We will obtain a solution of \eqref{eqLimit_sigma} by letting $a$
  tend to $+\infty$. We set $\gamma_a \definedas a^2$ and rescale
  $\phi_a$ and $W_a$ as
  \begin{equation*}
    \phitil_a \definedas \gamma_a^{-\frac{1}{2N}} \ph_a, \quad
    \Wtil_a \definedas \gamma_a^{-\frac{1}{2}} W_a. 
  \end{equation*}
  Equations \eqref{hamiltoniana}-\eqref{momentuma} can be rewritten as
  \begin{subequations}
    \begin{align}
      \frac{1}{\gamma_a^{\frac{1}{n}}} \left( \frac{4(n-1)}{n-2}
        \Delta \phitil_a + R \phitil_a \right) &= -\frac{n-1}{n}
      \tau^2 \phitil_a^{N-1} + |\sigma + L\Wtil_a|^2 \phitil_a^{-N-1},
      \label{hamiltoniana_rescaled}
      \\
      -\frac{1}{2} L^* L \Wtil_a &= \frac{n-1}{n} \phitil_a^N d\tau .
      \label{momentuma_rescaled}
    \end{align}
  \end{subequations}
  We divide the argument in two cases.

  \begin{case} Assume that
    \begin{equation*}
      \limsup_{a\to \infty} \int_M \left|L\Wtil_a\right|^2 \, dv  
      = + \infty. 
    \end{equation*}
  \end{case}
  This situation is very similar to the situation in the proof of
  Theorem \ref{main} when the energy $\gamma(\phi,W)$ is unbounded.
  Setting
  \begin{equation*}
    \gamma'_a \definedas 
    \max \left\{ \int_M \left|L\Wtil_a\right|^2 \, dv, 1 \right\}
  \end{equation*}
  and defining $\Wtil'_a \definedas (\gamma'_a)^{-\frac{1}{2}}
  \Wtil_a$, one sees that $\Wtil'_a$ converges to a non-trivial
  solution of the limit equation \eqref{eqLimit1}. Since this equation
  is assumed to have no solution except $0$, the first case cannot
  occur.

  \begin{case} Assume that
    \begin{equation*}
      \limsup_{a\to \infty} \int_M \left|L\Wtil_a\right|^2 \, dv
      < +\infty.
    \end{equation*}
  \end{case}
  Since $\gamma_a$ tends to $+\infty$, the situation is again similar
  to the situation in the proof of Theorem \ref{main} when the energy
  $\gamma(\phi,W)$ is unbounded. The difference here is that the
  $\sigma$-term does not vanish in the limit. So, mimicking the proof
  of Theorem \ref{main}, one sees that $\Wtil_a$ tends to a solution
  of Equation \eqref{eqLimit_sigma}. Note that in Equation
  \eqref{eqLimit_sigma} should appear some $\al_0 \in (0,1]$, but
  Remark \ref{epi=0} implies that $\al_0=1$.
\end{proof}

We can now give an example showing that non-existence of solutions of
the limit equation is a property which is not open in the
$C^2$-topology.

\begin{proof}[Proof of Proposition \ref{example_nd}] 
  On $S^n = \{(x_1,\dots,x_{n+1}) \in \mR^{n+1} \mid \sum_{i=1}^{n+1}
  x_1^2 = 1 \}$, we choose a fixed open neighborhood $V$ of the north
  pole $N \definedas (1,0,\dots,0)$ and the south pole $S\definedas
  (-1,0,\dots,0)$. We choose a metric $g$ which has no local conformal
  Killing field on $M \setminus V$ and coincides with the standard
  round metric in a neighborhood of the points $N$ and $S$. The
  results in \cite{BCS05} provide the existence of such a metric. Set
  $\tau \definedas e^{x_1}$. Note that $\tau > 0$ and that the
  gradient of $\ln(\tau)$ is a conformal Killing field of $S^n$
  equipped with its standard round metric. Further, the critical
  points of $\tau$ are $N$ and $S$. By construction of the metric $g$
  we have $L( d\ln(\tau) ) = L \left( d \tau / \tau \right) = 0$ on a
  neighborhood of $N$ and $S$. This implies that there exists a
  constant $C >0$ such that for all $x \in S^n$
  \begin{equation} \label{taultau} \left| L \left( \frac{d \tau}{\tau}
      \right)\right| \leq C \left| \frac{d\tau}{\tau} \right|^2
  \end{equation}
  everywhere on $S^n$.

  We will show that $g$ together with $\tau^a$ for $a>0$ large enough
  satisfy the conclusion of Proposition \ref{example_nd}. To do this
  we proceed by contradiction and assume that for every $a >0$, the limit
  equation \eqref{eqLimit1} has no non-trivial solution for any 
 $\al_0 \in (0,1]$. By
  Proposition \ref{limit_equations}, this means that for all $a>0$ and for 
  all trace-free and divergence-free $(0,2)$-tensor
  $\sigma_a$ there exists a non-trivial solution $W_a \in W^{2,p}$ of
  \begin{equation} \label{eqa} -\frac{1}{2} L^* L W_a =
    \sqrt{\frac{n-1}{n}} | \sigma_a + LW_a| \frac{d
      \tau^a}{\tau^a} = \sqrt{\frac{n-1}{n}} a | \sigma_a + LW_a|
    \frac{d \tau}{\tau}
  \end{equation}  

  Now choose a fixed non-zero  trace-free and divergence-free
  $(0,2)$-tensors $\sigma$ supported in $M \setminus V$. 
For the existence of $\sigma$
  with such properties, see for example \cite{Delay}.  We choose
  $\sigma_a = a^{-1} \sigma$ and get from Equation \eqref{eqa} a
  non-trivial solution of
  \begin{equation*}
    -\frac{1}{2} L^* L W_a 
    = \sqrt{\frac{n-1}{n}} |\sigma + a LW_a| \frac{d \tau}{\tau}. 
  \end{equation*} 
  Take the scalar product of this equation with $d\tau/\tau$ and
  integrate. This gives
  \begin{equation*}
    \sqrt{\frac{n-1}{n}} 
    \int_{S^n} |\sigma +a LW_a| 
    \left|\frac{d \tau}{\tau} \right|_g^2 \, dv^{g} 
    = 
    - \frac{1}{2}\int_{S^n} \< LW_a, L( d \tau /\tau) \>  
    \, dv^{g}.
\end{equation*} 
  Together with Inequality
  \eqref{taultau},  we obtain
  \begin{equation} \label{ineqmain} \int_{S^n} |\sigma +a LW_a|
    \left|\frac{d \tau}{\tau}\right|^2 \, dv^{g} \leq C
    \int_{S^n} \left|\frac{d \tau}{\tau}\right|^2 |LW_a| \,
    dv^{g},
  \end{equation}
  with maybe a new value of $C$. Since $|\sigma + a LW_a| \geq
  a|LW_a| - |\sigma|$ we conclude from Inequality \eqref{ineqmain}
  that
  \begin{equation*}
    (a - C)  \int_{S^n}  
    \left|\frac{d \tau}{\tau}\right|_g^2 |LW_a| \, dv^{g}
    \leq 
    \int_{S^n} |\sigma|  
    \left|\frac{d \tau}{\tau}\right|_g^2 \, dv^{g}. 
  \end{equation*}  
 The right hand side here is bounded which shows that
  \begin{equation*}
    \lim_{a\to + \infty}  \int_{S^n}  
    \left| \frac{d \tau}{\tau} \right|_g^2 |LW_a| \, dv^{g} = 0,
  \end{equation*} 
  Inequality \eqref{ineqmain} then tells us that
  \begin{equation} \label{limlhs1} \lim_{a \to \infty} \int_{S^n}
    |\sigma +a LW_a| \left|\frac{d \tau}{\tau}\right|_g^2 \,
    dv^{g} = 0.
  \end{equation} 
  Since the critical set of $\tau$ is exactly $N$ and $S$ we have
  \begin{equation*}
    \left|\frac{d \tau}{\tau}\right|_g^2 \geq \epsilon > 0
  \end{equation*} 
  on $S^n \setminus V$. From \eqref{limlhs1} we obtain
  \begin{equation} \label{limlhs2} \lim_{a \to \infty} \int_{S^n
      \setminus V} |\sigma +a LW_a| \, dv^{g} = 0.
  \end{equation} 
  Since $\sigma$ is supported in $S^n \setminus V$, we have
  \begin{equation*}
    \left| \int_{S^n} \< \sigma, \sigma + a LW_a \> \, dv^{g}\right| 
    \leq 
    \left\|\sigma\right\|_{L^\infty} \int_{S^n \setminus V} |\sigma + a LW_a| \, dv^{g}. 
  \end{equation*} 
  Together with \eqref{limlhs2}, this shows that
  \begin{equation*}
    \liminf_{a \to \infty} \int_{S^n} \< \sigma, a LW_a \> \, dv^{g} 
    = 
    - \liminf_{a \to \infty} \int_{S^n} |\sigma|^2 \, dv^g.
  \end{equation*} 
  We get that
$$\liminf_{a \to \infty} \int_{S^n} \< \sigma, a LW_a \> \, dv^{g} 
\neq 0.$$ However, since $\sigma$ is divergence-free, we must have
\begin{equation*}
  \int_{S^n} \< \sigma, a LW_a \> \, dv^{g} = 0
\end{equation*} 
for all $a$, which gives the desired contradiction.
\end{proof}

\section{Proof of Theorem \ref{nonexist} }
\label{sec_proof_nonexist}

\begin{proof}[Proof of Theorem \ref{nonexist}]
  We proceed by contradiction and assume that the system
  \eqref{hamiltonian}-\eqref{momentum} admits a solution $(\phi, W)$
  where $\phi > 0$. Note that
  \begin{equation*}
    \int_M \phi^{N+1} \Delta \phi \, dv 
    = 
    \frac{N+1}{\left(\frac{N}{2}+1\right)^2} 
    \int_M \left|d \phi^{\frac{N}{2}+1}\right|^2 \, dv.
  \end{equation*}
  Hence, multiplying the Lichnerowicz equation \eqref{hamiltonian} by
  $\phi^{N+1}$ and integrating over $M$ we get
  \begin{equation*}
    \begin{split}
      &\frac{4(n-1)}{n-2} \frac{N+1}{\left(\frac{N}{2}+1\right)^2}
      \int_M \left|d \phi^{\frac{N}{2}+1}\right|^2 \, dv + \int_M R
      \phi^{N+2} \, dv +
      \frac{n-1}{n} \int_M \tau^2 \phi^{2N} \, dv \\
      &\quad = \int_M |LW|^2 \, dv,
    \end{split}
  \end{equation*}
  since $\si \equiv 0$. By assumption the sum of the first two terms
  is non-negative, and we obtain
  \begin{equation} \label{inegforphi} \frac{n-1}{n} \int_M \tau^2
    \phi^{2N} \, dv \leq \int_M |LW|^2 \, dv.
  \end{equation} 
  Next, we take the scalar product of the vector equation
  \eqref{momentum} with $W$ and integrate over $M$. Using the H\"older
  inequality we get
  \begin{equation*}
    \begin{split} 
      \frac{1}{2} \int_M |LW|^2 \, dv &=
      \frac{n-1}{n} \int_M \phi^N \langle d\tau , W \rangle \, dv \\
      &\leq \frac{n-1}{n} {\left( \int_M \tau^2 \phi^{2N} \, dv
        \right)}^{1/2} {\left( \int_M \left|\frac{d\tau}{\tau}
          \right|^n \, dv \right)}^{\frac{1}{n}} {\left( \int_M |W|^N
          \, dv \right)}^{\frac{1}{N}} .
    \end{split}
  \end{equation*}
  Using \eqref{inegforphi} and the definition of $C_g$ we obtain
  \begin{equation*}
    \frac{1}{2} C_g 
    \leq 
    \frac{1}{2} \frac{{\left( \int_M |LW|^2 \, dv \right)}^{\frac{1}{2}}}
    {{\left( \int_ M |W|^N \, dv \right)}^{\frac{1}{N}}}  
    \leq 
    \sqrt{\frac{n-1}{n}}
    \left\| \frac{d \tau}{\tau} \right\|_{L^n},
  \end{equation*}
  since $W \not\equiv 0$. This contradicts the assumption and proves
  Theorem \ref{nonexist}.
\end{proof}

\appendix
\section{Positivity of $C_g$}

\begin{lemma} \label{lem_Cg_positive} Suppose that $(M,g)$ has no
  conformal Killing vector fields, then $C_g > 0$.
\end{lemma}

\begin{proof}
  Since $M$ is compact there is a constant $\lambda$ such that $\Ric
  \leq \lambda g$. From Formula \eqref{LW2} we get
  \begin{equation*}
    \frac{1}{2} \int_M \left|LV\right|^2 \, dv 
    \geq 
    \int_M \left| \nabla V \right|^2 \, dv 
    - 
    \lambda \int_M \left|V\right|^2 \, dv
  \end{equation*}
  and
  \begin{equation}\label{ineg1}
    \frac{1}{2} \int_M \left|LV\right|^2 \, dv 
    \geq 
    \left\| V \right\|_{H^1}^2 
    - (\lambda + 1) \left\| V \right\|_{L^2}^2. 
  \end{equation}
  We argue by contradiction and assume that for all positive integers
  $k$ there exists $V_k \in H^1$ such that $\left\| V_k \right\|_{L^N}
  = 1$ and $\frac{1}{2} \int_M \left| LV_k \right|^2 \, dv \leq 1/k $.
  Since $M$ is compact it follows that $V_k$ is uniformly bounded in
  $L^2$, and {}from Equation \eqref{ineg1} we conclude that the
  sequence $V_k$ is uniformly bounded in $H^1$. By the compactness of
  the embedding $H^1 \to L^2$, we can assume that the sequence $V_k$
  converges to some $V_{\infty}$ in the $L^2$-norm.  By continuity of
  the norm $V_{\infty}$ satisfies $\left\|V_{\infty} \right\|_{L^2} =
  1$, and in particular, $V_{\infty} \not\equiv 0$.  For any smooth
  1-form $\xi$ we have
  \begin{equation*}
    \begin{split}
      \frac{1}{2} \left| \int_M \langle L^*L \xi, V_{\infty}\rangle \,
        dv \right| &= \frac{1}{2} \lim_{k \to \infty} \left| \int_M
        \langle L^*L \xi, V_k\rangle \, dv
      \right| \\
      &= \frac{1}{2} \lim_{k \to \infty} \left| \int_M \langle L \xi,
        L V_k \rangle \, dv
      \right| \\
      &\leq \lim_{k \to \infty} \left(\frac{1}{2} \int_M |L \xi|^2 \,
        dv \right)^{\frac{1}{2}}
      \left(\frac{1}{2} \int_M |L V_k|^2 \, dv \right)^{\frac{1}{2}} \\
      &= 0.
    \end{split}
  \end{equation*}
  Hence $V_{\infty}$ satisfies the equation $L^* L V_{\infty} = 0$ in
  the sense of distributions, so is a nonzero conformal Killing vector
  field. This is a contradiction which proves that $C_g > 0$.
\end{proof}

\providecommand{\bysame}{\leavevmode\hbox to3em{\hrulefill}\thinspace}
\providecommand{\MR}{\relax\ifhmode\unskip\space\fi MR }
\providecommand{\MRhref}[2]{%
  \href{http://www.ams.org/mathscinet-getitem?mr=#1}{#2}
}
\providecommand{\href}[2]{#2}


\end{document}